\documentclass[a4paper,USenglish,autoref]{lipics-v2021}

\usepackage{algorithm}
\usepackage[noend]{algpseudocode}
\usepackage{amsmath,amssymb,amsthm}
\usepackage{booktabs}
\usepackage{caption}
\usepackage{enumitem}
\usepackage{graphicx}
\usepackage{hyperref}
\usepackage{pifont}  
\usepackage{subcaption} 
\usepackage{xcolor}

\newcommand{\bigo}[1]{\ensuremath{\mathcal{O}(#1)}}
\newcommand{\true}{\ensuremath{\textsc{true}}}
\newcommand{\false}{\ensuremath{\textsc{false}}}
\newcommand{\nil}{\ensuremath{\textsc{null}}}
\newcommand{\cmark}{\ding{51}}
\newcommand{\xmark}{\ding{55}}
\newcommand{\Gtri}{\ensuremath{G_{\Delta}}}
\newcommand{\Gfcc}{\ensuremath{G_{\mathrm{FCC}}}}
\newcommand{\Connected}{\ensuremath{\textsc{Connected}}}
\newcommand{\Read}{\ensuremath{\textsc{Read}}}
\newcommand{\Write}{\ensuremath{\textsc{Write}}}

\newcommand{\Lock}{\ensuremath{\textsc{Lock}}}

\newcommand{\candidate}{\ensuremath{\texttt{candidate}}}
\newcommand{\leader}{\ensuremath{\texttt{leader}}}
\newcommand{\nbrcand}{\ensuremath{\texttt{nbrcand}}}

\makeatletter
\algrenewcommand\ALG@beginalgorithmic{\small}
\algrenewcommand\alglinenumber[1]{\scriptsize #1:}
\makeatother

\algdef{SE}[SUBALG]{Indent}{EndIndent}{}{\algorithmicend\ }%
\algtext*{Indent}
\algtext*{EndIndent}
\makeatletter
\newcommand{\multiline}[1]{%
  \begin{tabularx}{\dimexpr\linewidth-\ALG@thistlm}[t]{@{}X@{}}
    #1
  \end{tabularx}
}
\makeatother

\newif\ifcomment
\commenttrue    

\newif\iffigabbrv
\figabbrvfalse   
\newcommand{\figtext}{\iffigabbrv Fig.\else Figure\fi}

\title{Asynchronous Deterministic Leader Election in Three-Dimensional Programmable Matter}
\titlerunning{Asynchronous Deterministic Leader Election in 3D Programmable Matter}

\author{Joseph L. Briones}{School of Computing and Augmented Intelligence, Arizona State University, Tempe, AZ, USA}{joseph.briones@asu.edu}{https://orcid.org/0000-0002-5847-4263}{NSF (CCF-2106917) and U.S.\ ARO (MURI W911NF-19-1-0233).}
\author{Tishya Chhabra}{Corona del Sol High School, Tempe, AZ, USA}{tishyac3.141@gmail.com}{https://orcid.org/0000-0002-3555-1078}{}
\author{Joshua J. Daymude}{Biodesign Center for Biocomputing, Security and Society,\\Arizona State University, Tempe, AZ, USA}{jdaymude@asu.edu}{https://orcid.org/0000-0001-7294-5626}{The Momental Foundation's Mistletoe Research Fellowship and the ASU Biodesign Institute.}
\author{Andr\'ea W. Richa}{School of Computing and Augmented Intelligence, Arizona State University, Tempe, AZ, USA}{aricha@asu.edu}{https://orcid.org/0000-0003-3592-3756}{NSF (CCF-1733680, CCF-2106917) and U.S.\ ARO (MURI W911NF-19-1-0233).}
\authorrunning{J.\ L.\ Briones, T.\ Chhabra, J.\ J.\ Daymude, and A.\ W.\ Richa}

\Copyright{Joseph L.\ Briones, Tishya Chhabra, Joshua J.\ Daymude, and Andr\'ea W.\ Richa}



\supplementdetails{Software}{https://github.com/SOPSLab/3DLeaderElection-Mathematica}

\hideLIPIcs
\nolinenumbers

\begin{document}

\maketitle

\begin{abstract}
    Over three decades of scientific endeavors to realize \textit{programmable matter}, a substance that can change its physical properties based on user input or responses to its environment, there have been many advances in both the engineering of modular robotic systems and the corresponding algorithmic theory of collective behavior.
    However, while the design of modular robots routinely addresses the challenges of realistic three-dimensional (3D) space, algorithmic theory remains largely focused on 2D abstractions such as planes and planar graphs.
    In this work, we formalize the \textit{3D geometric space variant} for the \textit{canonical amoebot model} of programmable matter, using the \textit{face-centered cubic (FCC) lattice} to represent space and define local spatial orientations.
    We then give a distributed algorithm for \textit{leader election} in connected, contractible 2D or 3D geometric amoebot systems that deterministically elects exactly one leader in $\bigo{n}$ rounds under an unfair \textit{sequential} adversary, where $n$ is the number of amoebots in the system.
    We then demonstrate how this algorithm can be transformed using the concurrency control framework for amoebot algorithms (DISC 2021) to obtain the first known amoebot algorithm, both in 2D and 3D space, to solve leader election under an unfair \textit{asynchronous} adversary.
\end{abstract}

\section{Introduction} \label{sec:intro}

Since its inception~\cite{Toffoli1991-programmablematter}, programmable matter has been envisioned as a material that can dynamically alter its physical properties based either on user input or autonomous sensing of the environment.
Many strides have been made to realize this technology over the last several decades, both from the practical perspective of modular robotics and the algorithmic contributions of distributed computing theory.
However, when it comes to realistic considerations of three-dimensional (3D), gravity-bound space, advances in robotics have outpaced their distributed computing counterparts.
Modular, reconfigurable robotic systems such as Proteo~\cite{Yim2001-distributedcontrol}, SlidingCube~\cite{Fitch2008-millionmodule}, 3D M-Blocks~\cite{Romanishin2015-3dmblocks}, RollingSphere~\cite{Luo2020-obstaclecrossingstrategy}, FireAnt3D~\cite{Swissler2020-fireant3d}, FreeBOT~\cite{Liang2020-freebotfreeform,Luo2022-adaptiveflow}, and 3D Catoms~\cite{Piranda2018-designingquasispherical,Thalamy2021-engineeringefficient} routinely address engineering challenges both in individual module design (such as binding and locomotion) and in collective reconfiguration (such as gravity stability) that are inherent to 3D environments.
Besides a few notable exceptions~\cite{Yamauchi2017-planeformation,Yamada2022-searchmetamorphic}, the vast majority of abstract models of mobile robots and programmable matter treat space as two-dimensional (2D) planes or planar graph structures~\cite{Chirikjian1994-kinematicsmetamorphic,DAngelo2020-asynchronoussilent,Derakhshandeh2014-amoebotba,Flocchini2019-distributedcomputing,Michail2016-simpleefficient,Patitz2014-introductiontilebased,Woods2013-activeselfassembly}, simplifying their assumptions but limiting their application to practical domains.

Our goal is to move theoretical programmable matter research towards the 3D reality by extending the established \textit{amoebot model} of programmable matter~\cite{Daymude2021-canonicalamoebot,Derakhshandeh2014-amoebotba}.
Research using the amoebot model has historically assumed 2D discretizations of space, most commonly the ``geometric'' triangular lattice (\figtext~\ref{fig:2dlattice}).
Under this treatment of space, amoebot algorithms have been developed for a myriad of problems including leader election, shape formation, object coating and enclosure, bridging, and more (see~\cite{Daymude2021-collaboratingmotion,Daymude2019-computingprogrammable,Daymude2021-canonicalamoebot} for an overview of results).
The recent \textit{canonical amoebot model}~\cite{Daymude2021-canonicalamoebot} systematized the many disparate assumptions appearing in these works into categories, each with a set of ``assumption variants'' of varying strengths.
In this paper, we formalize
the \textit{3D geometric space variant} for the canonical amoebot model that discretizes space as the \textit{face-centered cubic (FCC) lattice} (\figtext~\ref{fig:3dlattice}).

\begin{figure}[tb]
    \centering
    \begin{subfigure}{.47\textwidth}
    	\centering
    	\includegraphics[height=5cm]{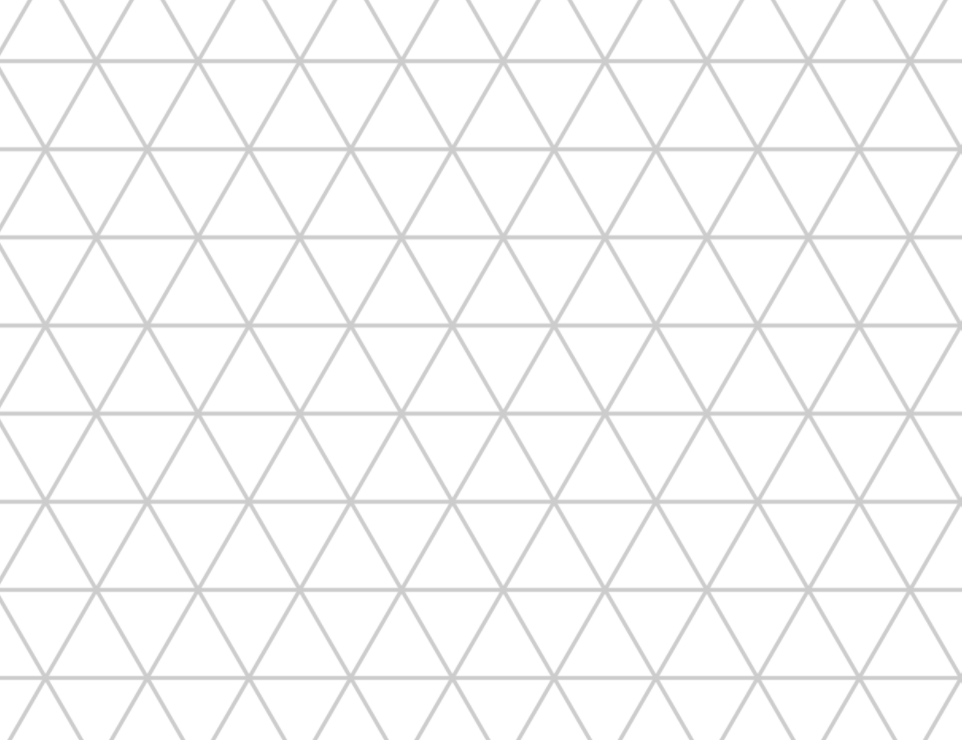}
    	\caption{\centering}
    	\label{fig:2dlattice}
    \end{subfigure} \hfill
    \begin{subfigure}{.47\textwidth}
    	\centering
    	\includegraphics[trim=40 25 30 40, clip, height=5cm]{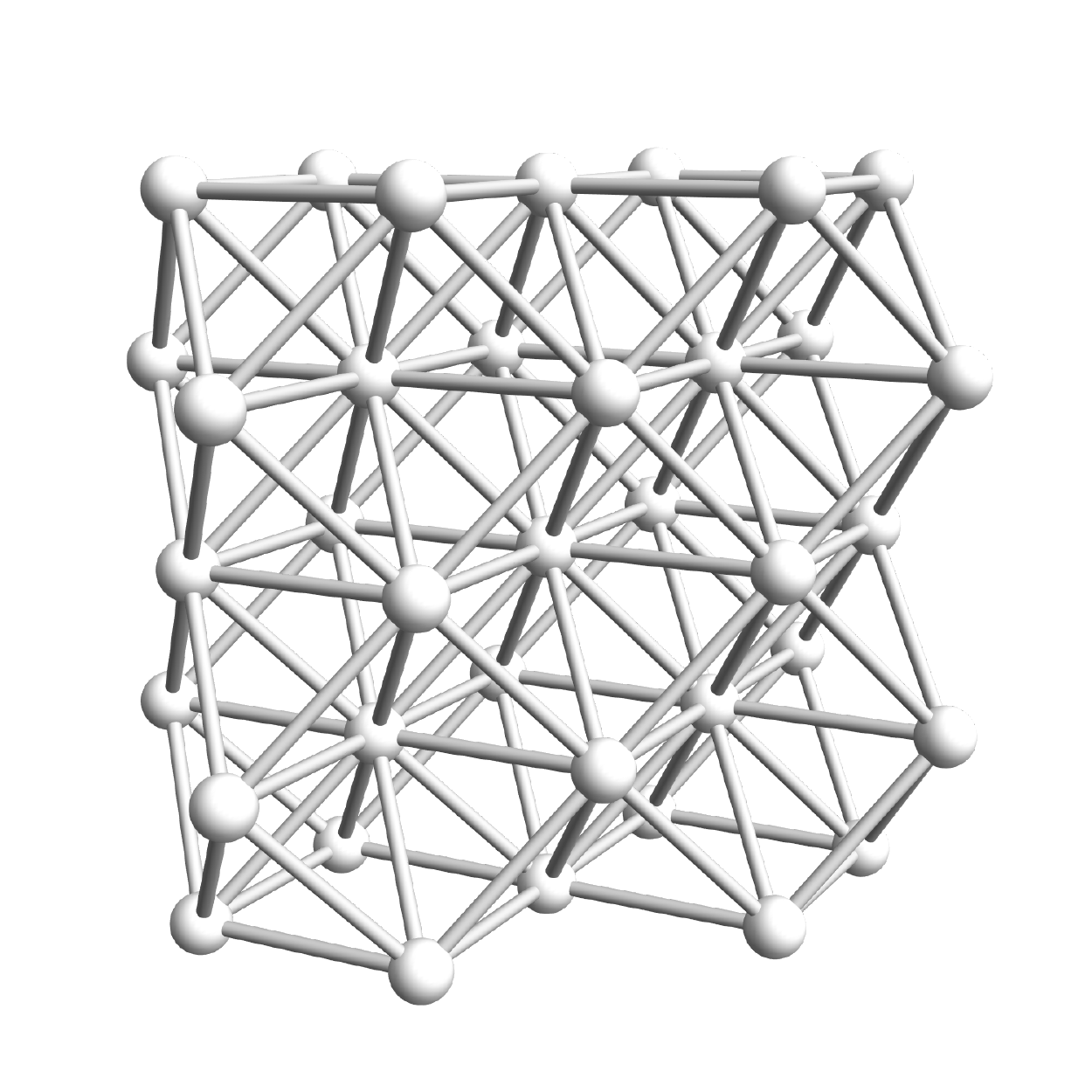}
    	\caption{\centering}
    	\label{fig:3dlattice}
    \end{subfigure}
    \caption{\textit{The Geometric Space Variants.}
    (a) The 2D triangular lattice $\Gtri$.
    (b) The 3D face-centered cubic (FCC) lattice $\Gfcc$.}
    \label{fig:lattices}
\end{figure}

Lattice representations of modular robots and programmable matter typically assume either a cubic lattice corresponding to cubic or spherical modules~\cite{Luo2020-obstaclecrossingstrategy,Luo2022-adaptiveflow,Romanishin2015-3dmblocks,Yamada2022-searchmetamorphic} or an FCC lattice corresponding to (quasi-)spherical or rhombic-dodecahedral modules~\cite{Gastineau2022-leaderelection,Piranda2018-designingquasispherical,Thalamy2021-engineeringefficient}.
The FCC lattice is a natural 3D generalization of the 2D triangular lattice already used for many amoebot algorithms; in fact, it can be decomposed into layers of triangular lattices and thus remains compatible with existing results for 2D amoebot systems.
Also, maintaining module connectivity during movements is easier in an FCC lattice than in a cubic one (see, e.g.,~\cite{Piranda2018-designingquasispherical}).

Using this new 3D geometric space variant, we revisit the classical problem of \textit{leader election}, defined formally in Section~\ref{sec:problem}.
We demonstrate that among the many existing algorithms for leader election in 2D amoebot systems~\cite{Bazzi2019-stationarydeterministic,Daymude2017-improvedleader,Derakhshandeh2015-leaderelection,DiLuna2020-shapeformation,Dufoulon2021-efficientdeterministic,Emek2019-deterministicleader,Gastineau2019-distributedleader}, the \textit{erosion-based algorithm} of Di Luna et al.~\cite{DiLuna2020-shapeformation} extends naturally to 3D---surprisingly, without any cost to runtime.
Although erosion-based election in 3D may seem simple at first glance, its analysis requires new, non-trivial topological arguments specific to the 3D setting.
Our algorithm elects exactly one leader in any connected, contractible amoebot system (defined formally in Section~\ref{sec:3dspace}) within $\bigo{n}$ rounds under an unfair sequential adversary, where $n$ is the number of amoebots in the system.
We thus achieve similar guarantees as the state-of-the-art algorithm for 3D leader election by Gastineau et al.~\cite{Gastineau2022-leaderelection}, with two important differences: (1) we consider all 24 possible amoebot orientations in 3D achievable by rotation or reflection while Gastineau et al.\ only consider eight, and (2) we break symmetry using local comparisons between neighbors' orientations while Gastineau et al.\ assume 2-neighborhood vision.
We further show that our algorithm is compatible with the \textit{concurrency control framework} for amoebot algorithms~\cite{Daymude2021-canonicalamoebot}, implying that it can be transformed into an algorithm with equivalent behavior that remains correct even under an unfair asynchronous adversary.

\subparagraph{Our Contributions.}

Our main contributions are summarized as:
\begin{itemize}
    \item A formalization of the \textit{3D geometric space variant} for the canonical amoebot model using the FCC lattice to discretize space and define amoebots' spatial orientations (Section~\ref{sec:3dspace}).
    
    \item A deterministic amoebot algorithm that solves leader election in both 2D and 3D geometric space for connected, contractible systems under an unfair sequential adversary within $\bigo{n}$ rounds, where $n$ is the number of amoebots in the system (Sections~\ref{sec:alg}--\ref{sec:analysis}).
    
    \item An application of the concurrency control framework for amoebot algorithms~\cite{Daymude2021-canonicalamoebot} that yields the first known amoebot algorithm, in both 2D and 3D space, to solve leader election under an unfair asynchronous adversary (Section~\ref{sec:async}).
\end{itemize}

\section{The Amoebot Model} \label{sec:model}

We begin by describing the features of the canonical amoebot model that will be used in this work; a deeper description of the model and its rationale can be found in~\cite{Daymude2021-canonicalamoebot}.
In the canonical amoebot model, programmable matter consists of individual, homogeneous computational elements called \textit{amoebots}.
The structure of an amoebot system is represented as a subgraph of an infinite, undirected graph $G = (V,E)$ where $V$ represents all relative positions an amoebot can occupy and $E$ represents all atomic movements an amoebot can make.
Each node in $V$ can be occupied by at most one amoebot at a time.
There are many possible assumption variants one could make about space; here, we consider the \textit{2D geometric} variant which assumes $G = \Gtri$, the triangular lattice (\figtext~\ref{fig:2dlattice}), and the presently introduced \textit{3D geometric} variant which assumes $G = \Gfcc$, the face-centered cubic lattice (\figtext~\ref{fig:3dlattice}).

In this work, all amoebots remain \textsc{contracted}, each occupying a single node in $V$; other works also consider \textsc{expanded} amoebots that occupy a pair of adjacent nodes in $V$.
Each amoebot keeps a collection of ports---one for each edge incident to the node it occupies---that are labeled consecutively according to its own local, persistent \textit{orientation}.
An amoebot's orientation is defined according to space variant-specific information; we define orientation for our lattices of interest in Section~\ref{sec:3dspace}.
Two amoebots occupying adjacent nodes are said to be \textit{neighbors}.
Although each amoebot is \textit{anonymous}, lacking a unique identifier, an amoebot can locally identify its neighbors using their port labels.
In particular, amoebots $A$ and $B$ connected via ports $p_A$ and $p_B$ are each assumed to know one another's orientations and labels for $p_A$ and $p_B$.

Each amoebot has memory whose size is a model variant; here we assume \textit{constant-size} memories.
An amoebot's memory consists of two parts: a persistent \textit{public memory} that is only accessible to an amoebot algorithm via communication operations (defined next), and a volatile \textit{private memory} that is directly accessible by amoebot algorithms for temporary variables, private computation, etc.
\textit{Operations} define the programming interface for amoebot algorithms to communicate, move, and control concurrency that are, in reality, implemented via message passing (see~\cite{Daymude2021-canonicalamoebot} for details).
Our algorithm for leader election only makes use of the communication operations \Connected, \Read, and \Write.
\begin{itemize}
    \item The \Connected\ operation tests the presence of neighbors.
    $\Connected(p)$ returns \true\ if and only if there is a neighbor connected via port $p$.
    
    \item The \Read\ and \Write\ operations exchange information in public memory.
    $\Read(p, x)$ issues a request to read the value of a variable $x$ in the public memory of the neighbor connected via port $p$ while $\Write(p, x, x_{val})$ issues a request to update its value to $x_{val}$.
    If $p = \bot$, an amoebot's own public memory is accessed instead of a neighbor's.
\end{itemize}

Amoebot algorithms are defined as sets of \textit{actions}, each of the form $\langle label\rangle : \langle guard\rangle \to \langle operations\rangle$.
An action's \textit{label} specifies its name.
Its \textit{guard} is a Boolean predicate determining whether an amoebot $A$ can execute it based on the ports $A$ has connections on---i.e., which nodes adjacent to $A$ are (un)occupied---and information from the public memories of $A$ and its neighbors.
An action is \textit{enabled} for an amoebot $A$ if its guard is true for $A$, and an amoebot is \textit{enabled} if it has at least one enabled action.
An action's \textit{operations} specify the finite sequence of operations and computation in private memory to perform if this action is executed.

An amoebot is \textit{active} if it is currently executing an action and is \textit{inactive} otherwise.
The model assumes an \textit{adversary} controls the timing of amoebot activations and the resulting action executions, whose \textit{concurrency} and \textit{fairness} are assumption variants.
In this work, we consider two concurrency variants: \textit{sequential}, in which at most one amoebot can be active at a time; and \textit{asynchronous}, in which any set of amoebots can be simultaneously active.
We consider the most general fairness variant: \textit{unfair}, in which the adversary may activate any enabled amoebot.

An amoebot algorithm's time complexity is evaluated in terms of \textit{rounds} representing the time for the slowest continuously enabled amoebot to execute a single action.
Let $t_i$ denote the time at which round $i \in \{0, 1, 2, \ldots\}$ starts, where $t_0 = 0$, and let $\mathcal{E}_i$ denote the set of amoebots that are enabled or already executing an action at time $t_i$.
Round $i$ completes at the earliest time $t_{i+1} > t_i$ by which every amoebot in $\mathcal{E}_i$ either completed an action execution or became disabled at some time in $(t_i, t_{i+1}]$.
Depending on the adversary's concurrency, action executions may span more than one round.

\section{The Three-Dimensional (3D) Geometric Space Variant} \label{sec:3dspace}

\begin{figure}[t]
    \centering
    \begin{subfigure}{.47\textwidth}
    	\centering
    	\includegraphics[width=0.9\textwidth]{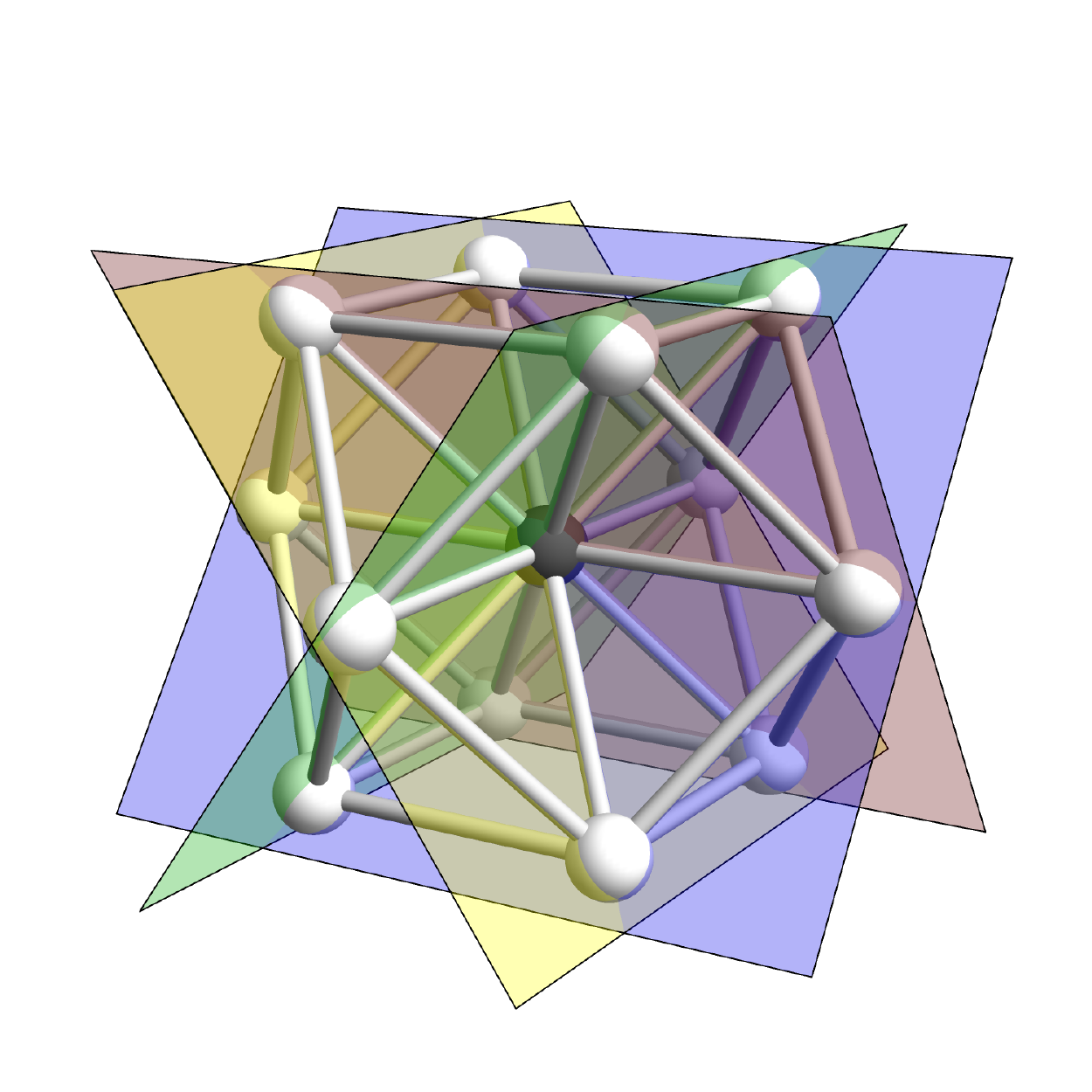}
    	\caption{\centering}
    	\label{fig:intersect:triangular}
    \end{subfigure} \hfill
    \begin{subfigure}{.47\textwidth}
    	\centering
    	\includegraphics[width=0.9\textwidth]{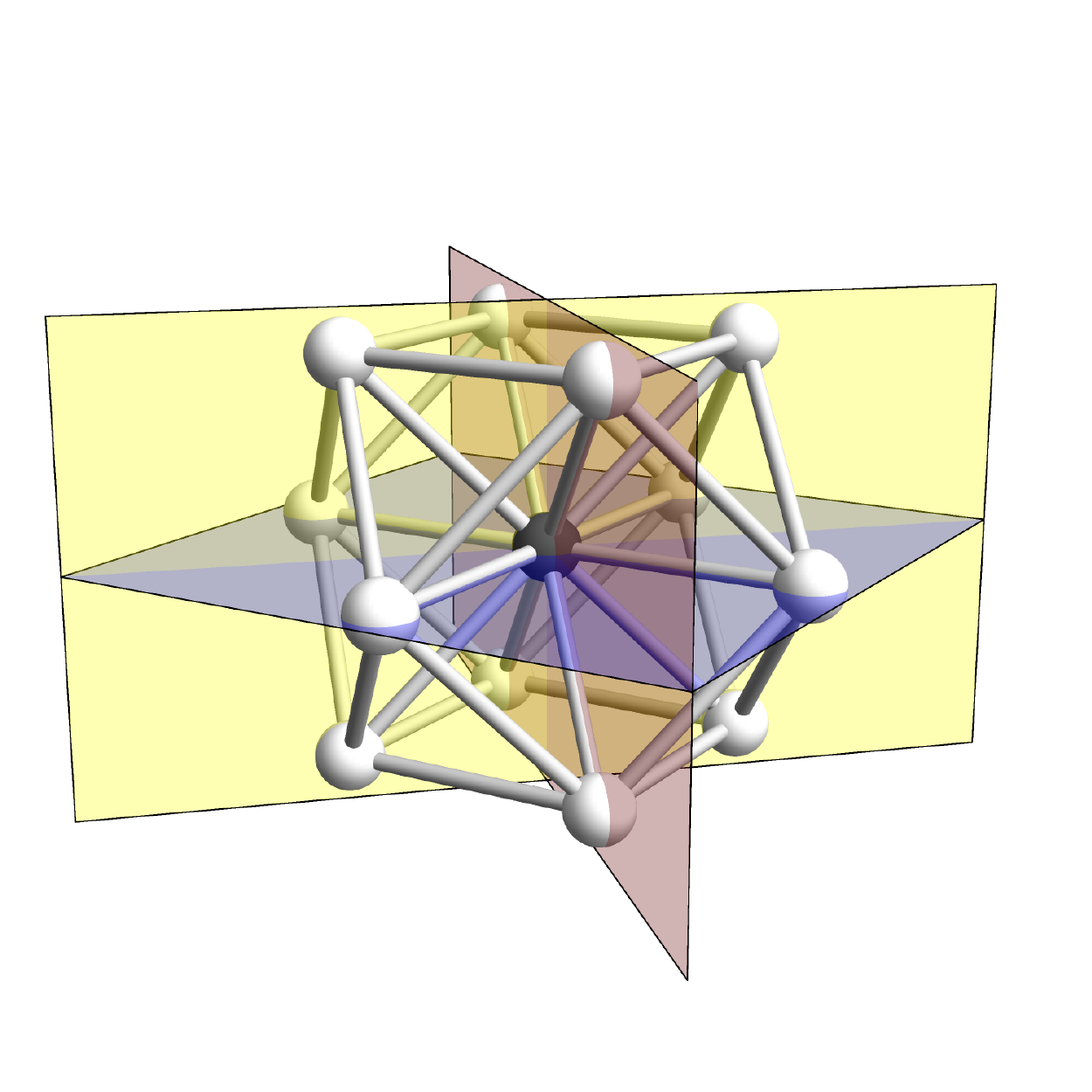}
    	\caption{\centering}
    	\label{fig:intersect:square}
    \end{subfigure}
    \caption{\textit{Lattice Decompositions of $\Gfcc$.}
    A node of $\Gfcc$ (black) viewed as the intersection of (a) four non-parallel triangular lattices or (b) three orthogonal square grids.}
    \label{fig:intersect}
\end{figure}

We now formally define the \textit{3D geometric space variant} for the canonical amoebot model, one of the main contributions of this work.
This variant assumes space is represented by the \textit{face-centered cubic (FCC) lattice}, $\Gfcc$ (\figtext~\ref{fig:3dlattice}).
Just as the triangular lattice $\Gtri$ used by the 2D geometric space variant can be viewed as the adjacency graph of the closest circle packing or as the dual of the hexagonal tiling, $\Gfcc$ can be viewed as the adjacency graph of the sphere packing that minimizes unoccupied volume or as the dual of the rhombic-dodecahedral tessellation.
Each node in $\Gfcc$ has degree 12 and can be viewed as the intersection of four infinite, non-parallel triangular lattices (\figtext~\ref{fig:intersect:triangular}) or as the intersection of three orthogonal square grids (\figtext~\ref{fig:intersect:square}).
Thus, in 3D geometric space, a contracted amoebot has 12 neighbors and an expanded amoebot has at most 18.

An amoebot's \textit{orientation} represents all the ways its local sense of space can be rotated or reflected while respecting the underlying spatial structure.
In $\Gtri$, an amoebot's orientation is represented as a \textit{direction} indicating which incident lattice edge it thinks of as ``north'' and a \textit{chirality} establishing the clockwise vs.\ counterclockwise ordering of its incident edges.
We generalize orientation in $\Gfcc$ using \textit{view}, \textit{spin}, and \textit{rotation} as defined below.
Different model variants may assume that amoebots share all, some, or none of their views, spins, and rotations in common.
This work assumes \textit{assorted} orientations, meaning the amoebots are not guaranteed to share any aspect of their orientations.

The \textit{home lattice} of an amoebot is one of the four infinite triangular lattices that contain the node the amoebot occupies (\figtext~\ref{fig:orientation}a).
An amoebot's \textit{view} is the triangular lattice decomposition of $\Gfcc$ containing the amoebot's home lattice (\figtext~\ref{fig:orientation}b).
A view can alternatively be defined as the set of triangular lattices whose planar embeddings in 3D space are non-intersecting, each orthogonal to the same two anti-parallel vectors.
An amoebot's \textit{spin} defines the ``top'' and ``bottom'' sides of the amoebot's home lattice; formally, it is the differentiation of the home lattice's two orthogonal vectors as positive and negative (\figtext~\ref{fig:orientation}c).
Having fixed this spin vector, any node's incident edges contained in the amoebot's view can be listed in clockwise order according to the right-hand rule.
The final component of an amoebot's orientation is its \textit{rotation} about its spin vector, of which there are three that agree with $\Gfcc$ (\figtext~\ref{fig:orientation}d).
A rotation can alternatively be defined as one of the three axes in $\Gfcc$ that contain the amoebot's node but are not contained in the amoebot's home lattice.\footnote{An amoebot's orientation could be defined more succinctly as a pair of vectors originating at the amoebot's node: one that is orthogonal to the amoebot's home lattice and pointing to its ``top'' side (defining view and spin), and a second that points to a neighboring node outside the amoebot's home lattice (defining rotation).}

\begin{figure}[t]
    \centering
    \includegraphics[width=0.9\textwidth]{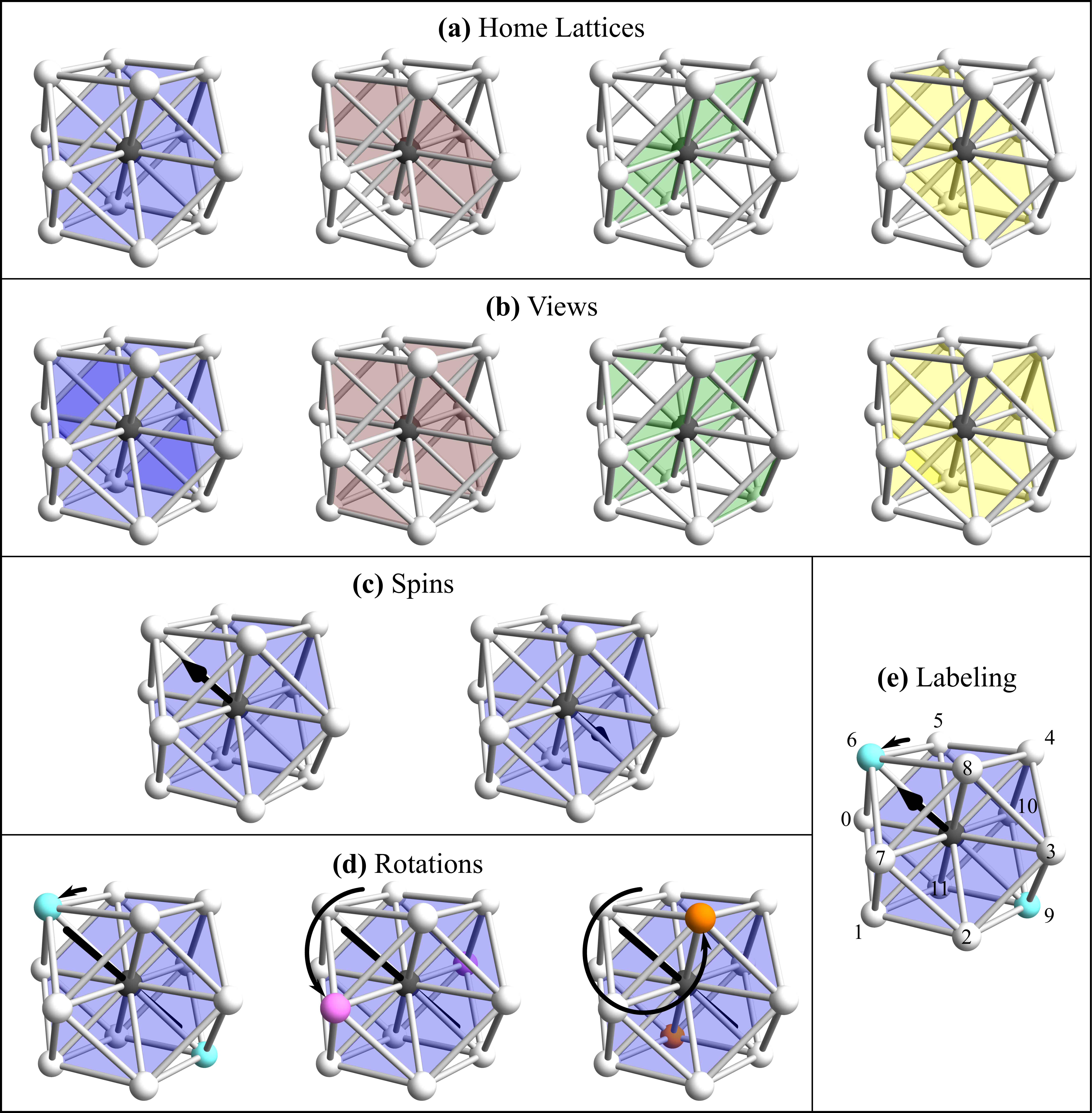}
    \caption{\textit{Amoebot Orientations.}
    (a)--(d) The components of an amoebot's orientation.
    For brevity, only the spins and rotations of the blue home lattice are shown in (c) and (d).
    (e) An example port labeling.}
    \label{fig:orientation}
\end{figure}

As in the 2D geometric space variant, an amoebot uses its orientation in 3D geometric space to define a consistent labeling of its ports (i.e., incident edges).
In \figtext~\ref{fig:orientation}e, we illustrate one possible labeling for a contracted amoebot.
The precise labeling scheme is immaterial so long as all amoebots label their ports consistently as a function of their (potentially differing) orientations.\footnote{This is strictly more general than the port labeling schemes allowed in the formulation of 3D space by Gastineau et al. (see Sections 2.1 and 2.2.1 of~\cite{Gastineau2022-leaderelection}). Their formulation assumes all amoebots share a common sense of ``layers'' analogous to our views, but for a square lattice decomposition of $\Gfcc$ instead of a triangular one.
With square ``home layers'', there are two possible spins (``up'' vs.\ ``down'') and four possible rotations, yielding eight total orientations.
Our formulation considers all 24.}

We can connect the 2D and 3D geometric space variants as follows.
The single triangular lattice $\Gtri$ in 2D geometric space can be thought of as any of the triangular lattices contained in $\Gfcc$; i.e., all amoebots in a 2D system have the same home lattice (and thus the same view) in a 3D system.
An amoebot's chirality in a 2D system plays the same role as its spin in a 3D system, defining ``up'' vs.\ ``down'' and clockwise vs.\ counterclockwise relative to the home lattice.
Finally, an amoebot's direction in a 2D system functions analogously to rotation about its spin vector in a 3D system, with one discrepancy: due to the constraints of the underlying lattices, a 2D system allows six possible $60^\circ$ rotations while a 3D system allows three possible $120^\circ$ rotations.

Finally, we define two spatial properties of amoebot systems that will be used throughout the remaining sections.
An amoebot system is \textit{connected} if the lattice nodes occupied by amoebots induce a single connected component.
We are also concerned with the notion of ``holes'' in amoebot systems.
In general topology, a \textit{hole} is any topological structure that prevents an object (or space) from being continuously shrunk to a single point; an object is \textit{contractible} if it does not contain holes.
To apply this definition to amoebot systems, recall that the dual of $\Gfcc$ (resp., $\Gtri$) is the rhombic dodecahedral (resp., hexagonal) tessellation of space.
Given a 3D (resp., 2D) amoebot system, its \textit{lattice dual representation} is the closed union of all solid rhombic dodecahedrons (resp., hexagons) in the lattice's dual corresponding to lattice nodes occupied by amoebots.
An amoebot system contains a hole if and only if its lattice dual representation contains a (topological) hole and is contractible otherwise.

\section{Amoebot Leader Election} \label{sec:problem}

\begin{table}[tb]
    \centering
    \caption{\textit{Comparison of Amoebot Algorithms for Leader Election.}
    Algorithms are organized chronologically by first publication.
    Gastineau et al.~\cite{Gastineau2022-leaderelection} details two separate leader election algorithms; the row with ``\#2'' is specific to its ``Algorithm 2'' which is marked with a $^*$ to denote its non-standard assumption of 2-neighborhood vision.
    When $k \in \mathbb{Z}_+$ appears in the number of leaders elected, it refers to the amoebot system being $k$-symmetric.
    For the runtime bounds, $L^*$ denotes the length of the longest system boundary, $L$ denotes the length of the system's outer boundary, $D$ is the diameter of the system's initial configuration, and $n$ denotes the number of amoebots in the system.
    The runtime terms $r$ and $mtree$ are specific to~\cite{Gastineau2019-distributedleader}, and it can be shown that $r + mtree$ is both $\Omega(D)$ and $\bigo{n}$.}
    \label{tab:leaderelection}
    \resizebox{\textwidth}{!}{%
    \begin{tabular}{llm{1.9cm}llm{1.1cm}m{1cm}m{1.3cm}l}
        \toprule
        \textbf{Algorithm} & \textbf{Space} & \textbf{Assorted Orientation} & \textbf{Adversary} & \textbf{Det.} &  \textbf{Allows ``Holes''} & \textbf{Statio- nary} & \textbf{Leaders Elected} & \textbf{Runtime} \\
        \midrule
        Derakhshandeh et al.~\cite{Derakhshandeh2015-leaderelection} & 2D & Direction & Strong Seq. & \xmark & \cmark & \cmark & 1 & $\bigo{L^*}$ exp. \\
        Daymude et al.~\cite{Daymude2017-improvedleader} & 2D & Direction & Strong Seq. & \xmark & \cmark & \cmark & 1, w.h.p. & $\bigo{L}$ w.h.p. \\
        Di Luna et al.~\cite{DiLuna2020-shapeformation} & 2D & Both & Sync. & \cmark & \xmark & \cmark & $k \leq 3$ & $\bigo{n^2}$ \\
        Gastineau et al.~\cite{Gastineau2019-distributedleader} & 2D & Direction & Strong Seq. & \cmark & \xmark & \cmark & 1 & $\bigo{r + mtree}$ \\
        Emek et al.~\cite{Emek2019-deterministicleader} & 2D & Both & Strong Seq. & \cmark & \cmark & \xmark & 1 & $\bigo{Ln^2}$ \\
        Bazzi and Briones~\cite{Bazzi2019-stationarydeterministic} & 2D & Direction & Weak Seq. & \cmark & \cmark & \cmark & $k \leq 6$ & $\bigo{n^2}$ \\
        Dufoulon et al.~\cite{Dufoulon2021-efficientdeterministic} & 2D & Direction & Strong Seq. & \cmark & \cmark & \xmark & 1 & $\bigo{L + D}$ \\
        Gastineau et al.\ \#2$^*$~\cite{Gastineau2022-leaderelection} & 3D & Spin/Rotation & Strong Seq. & \cmark & \xmark & \cmark & 1 & $\bigo{n}$ \\
        \textbf{This Paper} & 2\&3D & All & Strong Seq. & \cmark & \xmark & \cmark & 1 & $\bigo{n}$ \\
        \textbf{This Paper} + \cite{Daymude2021-canonicalamoebot} & 2\&3D & All & Async. & \cmark & \xmark & \cmark & 1 & \textbf{?} \\
        \bottomrule
    \end{tabular}}
\end{table}

An algorithm \textit{solves the leader election problem} if for any connected system of initially contracted amoebots with well-initialized memories, eventually a single amoebot irreversibly declares itself the system's leader and no other amoebot ever does so.
A leader's ability to break symmetry and coordinate the system via broadcasts makes it a powerful primitive for other amoebot algorithms, spurring a flurry of recent research on amoebot algorithms for leader election~\cite{Bazzi2019-stationarydeterministic,DAngelo2020-asynchronoussilent,Daymude2017-improvedleader,Derakhshandeh2015-leaderelection,DiLuna2020-shapeformation,Dufoulon2021-efficientdeterministic,Emek2019-deterministicleader,Gastineau2019-distributedleader,Gastineau2022-leaderelection}.
Table~\ref{tab:leaderelection} compares these existing algorithms, their assumptions, and their outcomes to our own, though we treat the D'Angelo et al.\ algorithm~\cite{DAngelo2020-asynchronoussilent} and ``Algorithm 1'' of Gastineau et al.~\cite{Gastineau2022-leaderelection} separately due to their strong, one-off assumptions.
The common assumptions and outcomes are:
\begin{itemize}
    \item \textit{Space.} Nearly all amoebot algorithms for leader election assume the 2D geometric space variant.
    Recently, Gastineau et al.~\cite{Gastineau2022-leaderelection} introduced a pair of algorithms for amoebot leader election on $\Gfcc$, much like our own 3D geometric space variant, but with stronger assumptions (see below).
    
    \item \textit{Orientation.} Recall that in 2D, amoebot orientation is defined as a direction and chirality; analogously, 3D orientations are a view, spin, and rotation.
    In 2D, the algorithms of Di Luna et al.~\cite{DiLuna2020-shapeformation} and Emek et al.~\cite{Emek2019-deterministicleader} allow fully assorted orientations while the rest assumed common chirality.
    In 3D, ``Algorithm 2'' of Gastineau et al.~\cite{Gastineau2022-leaderelection} assumes common (square) views while our algorithm allows fully assorted orientations.
    
    \item \textit{Adversary/Scheduler.} Until the canonical amoebot model was introduced, most algorithms assumed a sequential scheduler under which at most one amoebot could be active at a time.
    Strong sequential schedulers allow each amoebot to read, write, and move in one atomic action.
    Weak sequential schedulers used by Bazzi and Briones~\cite{Bazzi2019-stationarydeterministic} are more general, ensuring only that reading, writing, and moving are each individually atomic but may not necessarily be combined into a single atomic action.
    Di Luna et al.~\cite{DiLuna2020-shapeformation} assume a synchronous scheduler that may activate arbitrary subsets of amoebots concurrently, but only in discrete stages.
    This paper starts with a strong sequential scheduler but ultimately considers the most general asynchronous adversary that allows for arbitrary concurrency among amoebot actions.
    
    \item \textit{Deterministic vs.\ Randomized.} Randomization is a classical technique for symmetry breaking, but incurs a failure probability (with respect to correctness, runtime, or both) that is not present in deterministic algorithms.
    The original Derakhshandeh et al.\ algorithm~\cite{Derakhshandeh2015-leaderelection} and its improvement by Daymude et al.~\cite{Daymude2017-improvedleader} are both randomized while the rest, including our present algorithm, are deterministic.
    
    \item \textit{Connectivity and Holes.} All existing algorithms assume connected initial system configurations.
    The algorithms of Di Luna et al.~\cite{DiLuna2020-shapeformation} and Gastineau et al.~\cite{Gastineau2019-distributedleader} further assume their 2D initial configurations are hole-free.
    Our algorithm analogously assumes its 3D initial configurations are contractible (as defined in Section~\ref{sec:3dspace}), which is a necessary but likely insufficient condition for the ``electable'' configurations assumed by ``Algorithm 2'' of Gastineau et al.~\cite{Gastineau2022-leaderelection}.
    
    \item \textit{Movement.} Emek et al.~\cite{Emek2019-deterministicleader} and Dufoulon et al.~\cite{Dufoulon2021-efficientdeterministic} utilize amoebots' movement capabilities as a mechanism for symmetry breaking.
    All other algorithms, including ours, are \textit{stationary}, relying only on communication to elect a leader.
    
    \item \textit{Number of Leaders Elected.} Due to symmetry in the initial system configuration, some of the stationary deterministic algorithms elect a constant number of leaders instead of a unique one.
    In particular, for $k$-symmetric system configurations, the Di Luna et al.\ algorithm~\cite{DiLuna2020-shapeformation} elects $k \in \{1, 2, 3\}$ leaders and the Bazzi and Briones algorithm~\cite{Bazzi2019-stationarydeterministic} elects $k \in \{1, 2, 3, 6\}$.
    All other algorithms, including ours, elect a unique leader.
    
    \item \textit{Runtime.} All existing algorithms' runtime bounds are given in sequential rounds, where a round ends when each amoebot has been activated at least once.
    Our analysis uses a comparable but more technical definition of a round that extends to any concurrency assumption and focuses on the behavior of enabled amoebots (see Section~\ref{sec:model}).
    Among the deterministic algorithms, only the Dufoulon et al.\ algorithm~\cite{Dufoulon2021-efficientdeterministic} achieves a faster runtime than ours.
\end{itemize}

D'Angelo et al.~\cite{DAngelo2020-asynchronoussilent} introduced the SILBOT model which is inspired by the amoebot model and aims to study what collective behaviors are possible when robots cannot communicate via messages or memory accesses.
Under this seemingly challenging model, they show how to deterministically elect up to three leaders (due to symmetry) if the 2D system configuration is connected and hole-free---and if the robots can sense the presence and shape of robots in their 2-neighborhood on $\Gtri$, a strong assumption that no other amoebot algorithm makes.
They then additionally show how the algorithm can be generalized to connected configurations with holes if each robot can sense which unoccupied nodes are holes and which belong to the outside of the system.
This very strong assumption is unique to SILBOT, resolving the challenging ``inner-outer boundary problem''~\cite{Daymude2017-improvedleader,Derakhshandeh2015-leaderelection} by simply granting the requisite knowledge to robots a priori.
Our algorithm makes no such assumptions.

The pair of Gastineau et al.\ algorithms for leader election in 3D amoebot systems are the most relevant to ours since they also use $\Gfcc$ to discretize space~\cite{Gastineau2022-leaderelection}.
However, their first algorithm assumes all amoebots have the same 3D orientation and $\bigo{n\log n}$ memory, where $n$ is the number of amoebots in the system.
Their second algorithm addresses some of these concerns by assuming orientations with assorted spins and rotations but common views and constant-size memory, but allows amoebots to view ``extended neighborhoods'' that include nodes at distance 2.
Our algorithm achieves the same leader election guarantees without these assumptions, using only constant-size memory and local comparisons between 1-neighbors' orientations for symmetry breaking.

\section{3D Leader Election by Erosion} \label{sec:alg}

Our algorithm for leader election in 3D amoebot systems is an extension of the ``lattice consumption'' algorithm for 2D systems by Di Luna et al.~\cite{DiLuna2020-shapeformation}.
In the lattice consumption algorithm, all amoebots are initially eligible for leader candidacy.
When activated, an eligible amoebot uses certain rules regarding the number and relative positions of its eligible neighbors to decide whether to \textit{erode}, revoking its candidacy without disconnecting the set of eligible amoebots or introducing a hole.
Assuming the initial configuration was connected and hole-free, Di Luna et al.~\cite{DiLuna2020-shapeformation} proved that under a synchronous scheduler, erosion would eventually reduce the system to 1, 2, or 3 candidate leaders depending on the system's symmetry.

Our algorithm generalizes this approach to the 3D geometric space variant by defining rules for erosion based on 3D neighborhoods.
It deterministically elects a unique leader for connected, contractible systems by leveraging the sequential adversary to break symmetry (Section~\ref{sec:analysis}).
We then lift this strong timing assumption to the asynchronous setting---the most general of all concurrency assumptions---using the lock-based concurrency control framework for amoebot algorithms~\cite{Daymude2021-canonicalamoebot} (Section~\ref{sec:async}).

\begin{algorithm}[tb]
    \caption{Leader Election by Erosion for Amoebot $A$} \label{alg:erosion}
    \begin{algorithmic}[1]
        \State $\textsc{Setup} : (A.\candidate = \nil) \to$
        \Indent
            \State \Write$(\bot, \candidate, \true)$.  \Comment{Become a candidate.}
            \For {each port $p$ of $A$}  \Comment{Inform neighbors of candidacy.}
                \If {\Connected$(p)$}
                    \State Let $p'$ be the neighbor's port connected to port $p$.
                    \State \Write$(p, \nbrcand(p'), \true)$.
                \EndIf
            \EndFor
        \EndIndent
        \State $\textsc{Erode} : (A.\candidate = \true) \wedge (\forall B \in N(A), B.\candidate \neq \nil) \wedge \Call{CanErode}{ } \to$
        \Indent
            \State $\Write(\bot, \candidate, \false)$.  \Comment{Revoke candidancy.}
            \For {each port $p$ of $A$}  \Comment{Inform neighbors of erosion.}
                \If {\Connected$(p)$}
                    \State Let $p'$ be the neighbor's port connected to port $p$.
                    \State \Write$(p, \nbrcand(p'), \false)$.
                \EndIf
            \EndFor
        \EndIndent
        \State $\textsc{DeclareLeader} : (A.\candidate = \true) \wedge (\forall B \in N(A), B.\candidate = \false) \to$
        \Indent
            \State $\Write(\bot, \leader, \true)$.
        \EndIndent
        \Function{CanErode}{ }
            \State \Return \true\ if and only if:
            \begin{itemize}[leftmargin=14mm]
                \item Rule~\ref{rule:one}: $A$ has exactly one candidate neighbor; or
                
                \item Rule~\ref{rule:twofiveconnected}: $A$ has two to five candidate neighbors and these neighbors' positions induce a connected subgraph; or
                
                \item Rule~\ref{rule:square}: $A$ has exactly two candidate neighbors that have a common candidate neighbor such that these four candidates induce a square in $\Gfcc$.
            \end{itemize}
        \EndFunction
    \end{algorithmic}
\end{algorithm}

Algorithm~\ref{alg:erosion} details our algorithm's pseudocode and Table~\ref{tab:variables} summarizes each amoebot's local variables.
Initially, the system has no leader and all amoebots exist in a special null candidacy state.
Every amoebot's first activation executes the \textsc{Setup} action which sets the amoebot as a candidate and informs its neighbors of its candidacy.
Once all neighbors of a given candidate amoebot $A$ have also done their setup actions, the \textsc{Erode} action becomes enabled for $A$ whenever $A$ satisfies an erosion rule.
When executing \textsc{Erode}, $A$ revokes its candidacy and informs its neighbors that it did so.
As we will prove in the next section, repeated executions of the \textsc{Erode} action eventually reduce the system to a single candidate that, upon finding no candidate neighbors, elects itself as leader in the \textsc{DeclareLeader} action.

\begin{table}[tb]
    \centering
    \caption{\textit{Algorithm Notation.}
    The domain, initialization, and description of the local variables used in the leader election algorithm by an amoebot $A$.}
    \label{tab:variables}
    \resizebox{\columnwidth}{!}{%
    \begin{tabular}{llll}
        \toprule
        \textbf{Variable} & \textbf{Domain} & \textbf{Init.} & \textbf{Description} \\
        \midrule
        \leader & $\{\true, \false\}$ & \false & \true\ iff $A$ is the unique leader \\
        \candidate & $\{\nil, \true, \false\}$ & \nil & After first activation, \true\ iff $A$ is a candidate \\
        \nbrcand$(p)$ & $\{\true, \false\}$ & \false & \true\ iff $A$ has a candidate neighbor on port $p$ \\
        \bottomrule
    \end{tabular}}
\end{table}

It remains to specify the erosion rules for 3D geometric systems.
We formally specify these rules below and visualize them in \figtext~\ref{fig:erosion}.
For the sake of clarity, we represent the collection of these rules in Algorithm~\ref{alg:erosion} as a function \textsc{CanErode} that returns \true\ if and only if the calling amoebot $A$ satisfies one of the following erosion rules:
\begin{enumerate}[label=\textbf{Rule \arabic*.}, ref=\arabic*, left=0mm]
    \item \label{rule:one} $A$ has exactly one candidate neighbor (\figtext~\ref{fig:erosion}a).
    
    \item \label{rule:twofiveconnected} $A$ has two to five candidate neighbors, and these neighbors' positions induce a connected subgraph (\figtext~\ref{fig:erosion}b).
    
    \item \label{rule:square} $A$ has exactly two candidate neighbors that have a common candidate neighbor such that these four candidates induce a square in $\Gfcc$ (\figtext~\ref{fig:erosion}c).
\end{enumerate}

As we will show in Section~\ref{sec:analysis}, any connected and contractible system of at least two candidate amoebots contains at least one candidate $A$ satisfying one of these three rules.
We can further show that the ``erosion'' of $A$ does not violate the connectivity or contractibility of the remaining candidate structure, ensuring that the system eventually converges to exactly one leader amoebot.

Rules~\ref{rule:one} and~\ref{rule:twofiveconnected} can be evaluated using only the local port labels of $A$ that are connected to candidate neighbors---which can be obtained using \Read\ operations on neighbors' \candidate\ variables---and some basic information about the structure of $\Gfcc$.
However, Rule~\ref{rule:square} covers a special case that is specific to $\Gfcc$.
If four candidates $A$, $B$, $C$, and $D$ form a square in which each candidate is adjacent to exactly two others\footnote{Adjacency in $\Gfcc$ is taken by (2D) face adjacency; note that amoebots $A,B,C$ and $D$ actually do intersect in the rhombic dodecahedral tessellation at exactly one vertex point.} but not to the third that is ``catty-corner'' to it, then none of these four candidates can determine whether erosion will disconnect the candidate structure when using only the positions of their candidate neighbors.
For $A$, safe erosion hinges on the existence of the catty-corner candidate $C$: if $A$ erodes, then $B$ and $D$ remain connected if and only if $C$ exists.

\begin{figure}[t]
    \centering
    \includegraphics[width=0.8\textwidth]{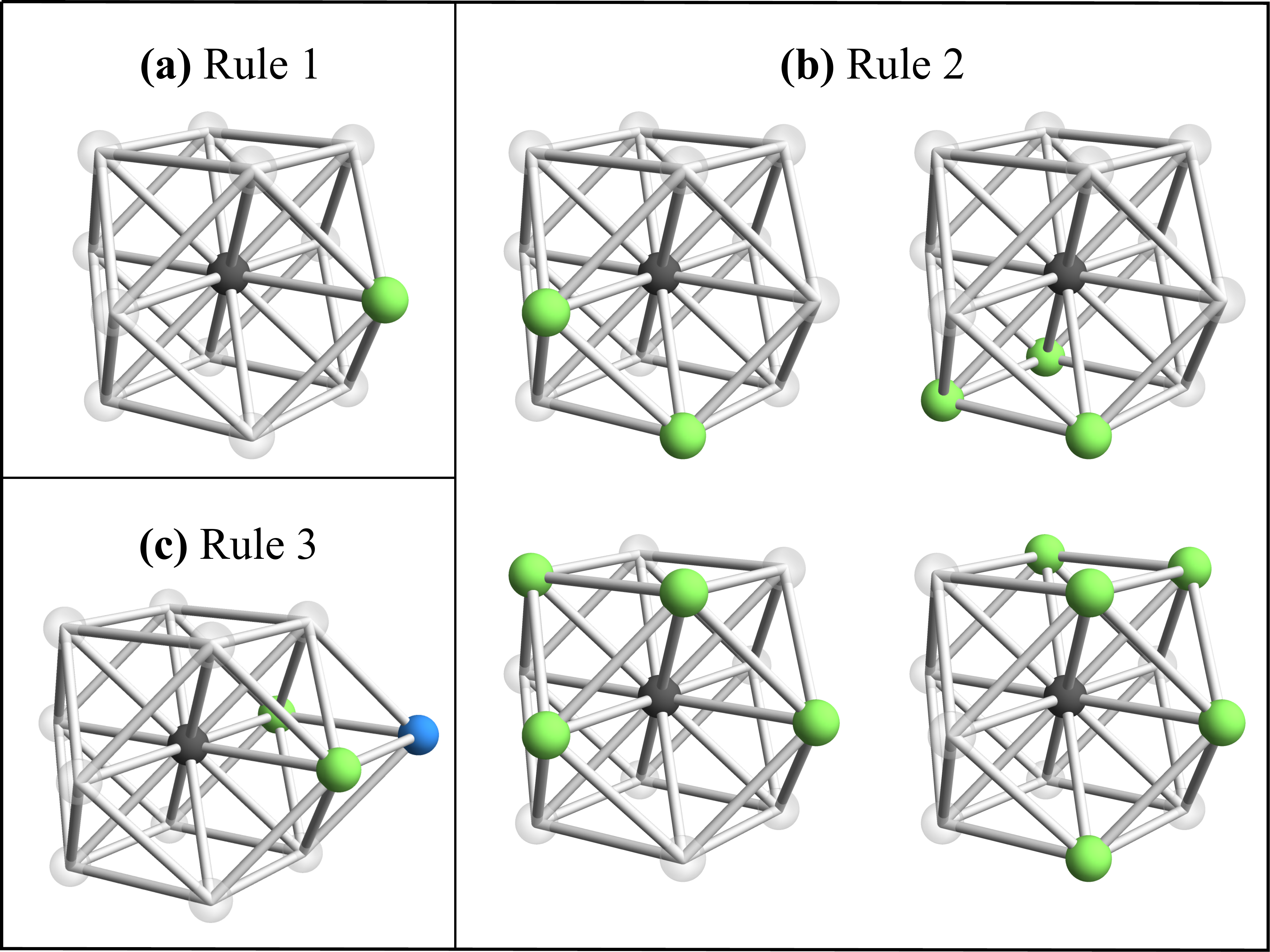}
    \caption{\textit{Erosion Rules.}
    Example configurations of neighboring candidates (green) for which an amoebot $A$ (black) can safely erode.
    (a)--(b) $A$ can erode if it has one to five neighbors that induce a connected subgraph.
    (c) $A$ can only erode if the catty-corner candidate neighbor (blue) exists.}
    \label{fig:erosion}
\end{figure}

Instead of assuming this problem away by giving the amoebots 2-neighborhood vision, as ``Algorithm 2'' of Gastineau et al.~\cite{Gastineau2022-leaderelection} does, we address this problem locally using the \nbrcand\ variables.
An amoebot $A$ has $A.\nbrcand(p) = \true$ if and only if $A$ has a neighbor candidate connected via port $p$; these are set to \true\ by neighbors becoming candidates during their \textsc{Setup} actions and then are reset to \false\ by eroding candidate neighbors during their \textsc{Erode} actions.
Maintaining these variables allows an amoebot $A$ that might be in a square to \Read\ the corresponding \nbrcand\ variable of its candidate neighbors $B$ or $D$ to check for the existence of the catty-corner candidate $C$.
Importantly, this uses the amoebot model's assumption that neighboring amoebots know one another's orientations, and thus $A$ can translate the adjacency it wants to check into the correct port label from the perspective of $B$ or $D$.

\section{Sequential Analysis} \label{sec:analysis}

We first analyze our algorithm under an unfair sequential adversary; i.e., when at most one amoebot can be active per time and the adversary may activate any enabled amoebot.
We will extend these results to an unfair asynchronous adversary allowing for arbitrary concurrency in the next section.
We first prove our algorithm never violates a safety condition.

\begin{lemma} \label{lem:safety}
    Let $S$ be a connected and contractible amoebot system on $\Gfcc$.
    If amoebot $A$ satisfies an erosion rule with respect to $S$, then $S - A$ is also connected and contractible.
\end{lemma}
\begin{proof}
    Suppose to the contrary that there exists an amoebot $A$ in $S$ that satisfies an erosion rule but $S - A$ is either disconnected or not contractible.
    Suppose first that $S - A$ is disconnected, i.e., there exist amoebots $\{B, C\} \not\ni A$ such that a simple path from $B$ to $C$ exists in the subgraph induced by $S$ but not in the subgraph induced by $S - A$.
    Any $(B,C)$-path $\mathcal{P}$ in $S$ must necessarily must have contained $A$ and, by extension, two distinct neighbors $N_1$ and $N_2$ of $A$.
    Thus, $A$ could not have eroded because it had exactly one neighbor in $S$ (Rule~\ref{rule:one}).
    Moreover, $A$ could not have eroded because it had between two and five neighbors in $S$ that themselves induced a connected subgraph (Rule~\ref{rule:twofiveconnected}) since the path $\mathcal{P} = (B, \ldots, N_1, A, N_2, \ldots, C)$ could be augmented using the connectivity of $N_1$ and $N_2$ as $(B, \ldots, N_1, \ldots, N_2, \ldots, C)$, a $(B,C)$-path in $S - A$.
    So $A$ must have eroded as a result of Rule~\ref{rule:square}; w.l.o.g., suppose the clockwise order of amoebots in the square is $(A, N_1, X, N_2)$.
    But $A$ could only have satisfied Rule~\ref{rule:square} if its catty-corner amoebot $X$ existed in $S$, implying the path $(B, \ldots, N_1, X, N_2, \ldots, C)$ exists in $S - A$.
    Therefore, in all cases, $S - A$ must be connected.
    
    So suppose instead that $S - A$ is connected but is not contractible; i.e., $S - A$ contains a hole.
    Since $S$ is contractible by supposition, the hole in $S - A$ must correspond to exactly the position of $A$; i.e., the lattice dual representation replacing the neighbors of $A$ with rhombic dodecahedrons must contain a topological hole.
    One can prove using exhaustive search that no configuration of at most five neighbors has a corresponding closed union of rhombic dodecahedrons containing a topological hole.\footnote{We verified this statement using a simple Mathematica program that enumerates each neighborhood polyhedron of at most five neighboring rhombic dodecahedrons and verifies that its genus is zero; see the Supplemental Material link at the start of this paper.}
    Since all erosion rules involve $A$ having between one and five neighbors in $S$, it is impossible for $A$ to both satisfy an erosion rule and create a hole by eroding, a contradiction.
\end{proof}

We next leverage this safety condition to prove progress.

\begin{lemma} \label{lem:progress}
    Any connected and contractible amoebot system $S$ on $\Gfcc$ comprising at least two amoebots contains an amoebot that satisfies an erosion rule with respect to $S$.
\end{lemma}
\begin{proof}
    Argue by induction on $|S|$, the number of amoebots in $S$.
    When $|S| = 2$, the fact that $S$ is connected implies that the two amoebots in $S$ each have each other as their only neighbor in $S$.
    Thus, both of these amoebots satisfy Rule~\ref{rule:one}.
    
    So suppose that $|S| > 2$ and that, for all $2 \leq k < |S|$, any connected and contractible amoebot system on $\Gfcc$ comprising $k$ amoebots contains an amoebot satisfying an erosion rule.
    We decompose the amoebots of $S$ into two types: ``chain'' amoebots whose removal would disconnect $S$, and those belonging to the remaining 2-connected components.
    A \textit{chain} is a maximal set of amoebots $A_1, \ldots, A_\ell$ in $S$ satisfying one of the following cases: (i) if $\ell > 2$, then for all $1 < i < \ell$, $(A_{i-1}, A_i)$ and $(A_i, A_{i+1})$ are the only edges in $S$ incident to $A_i$, $(A_1, A_\ell) \not\in S$, and no other chain can exist with endpoints $A_1$ and $A_\ell$, (ii) if $\ell = 2$, then $(A_1, A_2) \in S$ and no edge of $S$ exists between $\{A_1\} \cup N_S(A_1)$ and $\{A_2\} \cup N_S(A_2)$, or (iii) if $\ell = 1$, then $N_S(A_1)$ is disconnected.
    The final condition in (i) regarding the uniqueness of a chain with given endpoints follows from the squares that exist in $\Gfcc$.
    The lattice dual representation of a square is four rhombic dodecahedrons that intersect at a single point; thus, no one chain between catty-corner endpoints can individually disconnect $S$ and thus are omitted.
    No other such ``parallel chains'' can exist in $S$ because it is contractible.
    
    Consider the graph composed of (i) all chain amoebots $\mathcal{C}$ in $S$ and (ii) each connected component of $S - \mathcal{C}$ contracted to a single node.
    Since $S$ is contractible, this graph must be a tree.
    Consider any leaf of this tree.
    If this leaf corresponds to an amoebot $A^*$ with exactly one neighbor in $S$, then $A^*$ satisfies Rule~\ref{rule:one}.
    Otherwise, the leaf corresponds to a connected component $S'$ of $S - \mathcal{C}$ with at least two amoebots.
    Since $S'$ is a leaf in the tree, there exists a unique chain amoebot $C \in \mathcal{C}$ connecting $S'$ to the rest of $S$.
    
    The remainder of this proof identifies an amoebot $A^* \neq C$ in $S'$ that satisfies erosion Rule~\ref{rule:twofiveconnected} or \ref{rule:square}.
    If there exists an amoebot $A^*$ in $S'$ such that $N_{S' \cup \{C\}}(A^*) = \{X, Y\}$ and there exists an amoebot $Z$ such that $A^*$, $X$, $Y$, and $Z$ form an induced (chordless) square in $\Gfcc$, then $A^*$ satisfies erosion Rule~\ref{rule:square}.
    
    Otherwise, if an amoebot satisfying Rule~\ref{rule:square} cannot be found in $S'$, we show there must exist an amoebot $A^*$ in $S'$ satisfying Rule~\ref{rule:twofiveconnected}.
    We do so by constructing an auxiliary (not necessarily convex) polyhedron $P$ from the nodes of $\Gfcc$ occupied by amoebots in $S' \cup \{C\}$ that is contractible.
    Let $P$ be the polyhedron in the 3D embedding of \Gfcc\ that contains all the nodes occupied by $S' \cup \{C\}$ and whose faces are composed of (unions of) external triangles induced by $S' \cup \{C\}$, where an external triangle induced by $S'\cup \{C\}$ is a triangle whose three vertices are mutually adjacent nodes occupied by amoebots in $S' \cup \{C\}$ and where at least one of its vertices is adjacent to a node not occupied by $S' \cup \{C\}$.
    Note that the fact that the union of rhombic dodecahedrons corresponding to $S'\cup \{ C\}$ is contractible, which is true by assumption, does not immediately imply that $P$ will also be contractible, since by construction of $P$, $P$ might possibly contain an "exposed" chordless square cycle of $\Gfcc$ on its surface whose nodes do not share a common neighbor: This is not possible, however, since if such a cycle existed in $P$, there would exist an amoebot $A^*$ satisfying Rule~\ref{rule:square} -- i.e., an amoebot $A$ in $S'$ such that $N_{S'\cup \{ C\}}(A)=\{X,Y\}$ and there exists an amoebot $Z$ such that $A,X,Y,Z$ form an induced (chordless) square cycle in $\Gfcc$. 
    
    
    The \textit{internal angle} of a face $F$ at its vertex $v$ is given by the angle internal to $F$ defined by the two edges of $F$ that contain vertex $v$.
    The \textit{total internal angle} of a vertex $v$ is the sum of internal angles of all faces incident to $v$, and its \textit{external angle} is given by $2\pi$ minus its total internal angle:
    \[\text{external-angle}(v) = 2\pi - \sum_{i \in k}\alpha_i,\]
    where $\alpha_1, \ldots, \alpha_k$ are the interior angles of faces $F_1, \ldots, F_k$ incident to $v$.
    The \textit{Euler characteristic} on closed surfaces is given by
    \[\chi = 2 - 2g,\]
    where $g$ is the genus of the surface.
    Since $P$ was shown to be contractible, its genus is $g = 0$ and its Euler characteristic is $\chi = 2$.
    From the proof of the discrete version of the Gauss--Bonnet theorem \cite{Crane2013-digitalgeometry},
    \[\sum_{v \in P} \text{external-angle}(v) = 2\pi\chi.\]
    Thus, the sum of external angles of all vertices of $P$ is equal to $4\pi$.
    
    The external angle of any vertex can be at most $2\pi$, so there must exist at least two vertices of $P$ whose external angles are strictly positive.
    Since $C$ is the unique chain amoebot in $S' \cup \{C\}$, there must exist at least one non-chain amoebot $A^* \in S'$ at a vertex of $P$ with a strictly positive external angle, which by definition of a chain node must have a neighborhood in $S' \cup \{C\}$ that induces a connected subgraph.
    One can show via exhaustive search that any vertex of $P$ with a strictly positive external angle must correspond to a node with degree at most five in the embedded $\Gfcc$ lattice.\footnote{We verified this statement using a simple Mathematica program that enumerates all neighborhoods of six to eleven neighbors yielding a polyhedron vertex and verifies that its external angle is nonpositive; see the Supplemental Material link at the start of this paper.}
    Thus, $A^*$ must satisfy Rule~\ref{rule:twofiveconnected}.
   
    Therefore, there always exists an amoebot $A^*$ in $S$ satisfying an erosion rule.
    When $A^*$ erodes, $S - A^*$ remains connected and contractible by Lemma~\ref{lem:safety}, so the lemma follows by induction.
\end{proof}

The safety and progress lemmas yield our main theorem.

\begin{theorem} \label{thm:seq3d}
    Assuming 3D geometric space, assorted orientations, constant-size memory, and an unfair sequential adversary, Algorithm~\ref{alg:erosion} solves the leader election problem for any connected and contractible system of $n$ amoebots in $\bigo{n}$ rounds.
\end{theorem}
\begin{proof}
    We analyze the erosion process by applying Lemmas~\ref{lem:safety} and~\ref{lem:progress} to the structure of both null and real candidates $A$ with $A.\candidate \in \{\nil, \true\}$.
    At initialization, every amoebot is a null candidate, so the structure of candidates is connected and contractible by supposition.
    Lemma~\ref{lem:progress} guarantees that some candidate satisfies an erosion rule, but this does not immediately imply it can erode: it might only be a null candidate that has not yet executed \textsc{Setup}, or it might be a real candidate that has null candidate neighbors.
    In either case, the \textsc{Erode} action is disabled.
    
    Suppose that erosion progress is blocked, i.e., every candidate that satisfies an erosion rule---of which there must be at least one by Lemma~\ref{lem:progress}---has its \textsc{Erode} action disabled.
    Since the candidate structure is connected and contractible by Lemma~\ref{lem:safety}, the \textsc{DeclareLeader} action is never enabled for any amoebot until the very end, when all but one amoebot has $\candidate = \false$.
    This certainly cannot occur while there are still null candidates in the system.
    So even the unfair adversary has no choice but to activate some null candidate $A$ waiting to execute \textsc{Setup}, which is guaranteed to be continuously enabled for $A$ from initialization to its first activation.
    
    Thus, regardless of the unfair adversary's choice of activation, either a null candidate executes \textsc{Setup} to become a real candidate or a real candidate erodes by executing \textsc{Erode}.
    The guard of \textsc{Erode} ensures that any such erosion occurs after the candidate's neighbors have all executed their \textsc{Setup} actions, avoiding any discrepancies between the erosion rules dealing only with real candidate neighbors and this proof's inclusion of null candidates.
    Lemma~\ref{lem:safety} thus guarantees that any erosion maintains the connectivity and contractibility of the candidate structure, and Lemma~\ref{lem:progress} guarantees once again that some candidate in the remaining structure satisfies an erosion rule.
    Repeatedly applying this argument shows that the candidate structure is eventually reduced to a single candidate, for which \textsc{DeclareLeader} becomes enabled.
    This is the only enabled amoebot remaining in the system, so the unfair adversary must activate it, at which point the amoebot declares itself leader and the algorithm terminates.
    
    It remains to bound the worst-case runtime of our algorithm on a system of $n$ amoebots, i.e., the largest number of rounds that can complete before termination based on any possible activation sequence of the unfair adversary.
    As observed earlier, the \textsc{Setup} action is continuously enabled for each amoebot from initialization to its first activation, implying that the first round does not complete until all $n$ amoebots have executed their \textsc{Setup} actions, becoming (real) candidates.
    In the worst case, these are the only action executions that occur in the first round.
    At the start of the next round, Lemma~\ref{lem:progress} ensures there exists a candidate that satisfies an erosion rule; in the worst case, there is only one such candidate $A$.
    Because all other amoebots have moved past null candidacy, we are guaranteed that \textsc{Erode} is enabled for $A$.
    Thus, $A$ must be activated and will erode by the end of the second round.
    More generally, it can take at most $n - 1$ rounds for $n - 1$ candidates to erode and at most one additional round for the final candidate to declare itself leader.
    Therefore, our algorithm terminates within $\bigo{n}$ rounds in the worst case.
\end{proof}

Recall from Section~\ref{sec:3dspace} that $\Gtri$ used in 2D geometric space can be treated as a single triangular lattice in $\Gfcc$ used in 3D geometric space.
Certainly any connected and contractible 2D amoebot system is thus also a connected and contractible 3D system, yielding the following corollary.

\begin{corollary} \label{cor:seq2d}
    Theorem~\ref{thm:seq3d} also holds in 2D geometric space.
\end{corollary}

Notably, the special case involving erosion in squares of candidates never occurs in 2D geometric space, so an even simpler version of our algorithm that omits \nbrcand\ variables and Rule~\ref{rule:square} could be used to solve leader election in 2D.

\section{Asynchronous Leader Election} \label{sec:async}

The \textit{concurrency control framework} for amoebot algorithms uses the canonical amoebot model's \Lock\ operation to provide isolation among concurrent amoebot's actions~\cite{Daymude2021-canonicalamoebot}.
Specifically, the framework takes as input any amoebot algorithm that achieves a desired behavior under a sequential adversary and---provided the algorithm satisfies certain conventions---constructs another algorithm that produces equivalent behavior under an asynchronous adversary by carefully wrapping the original algorithm's actions in \Lock\ calls and reorganizing their operations.
In this section, we demonstrate that our algorithm for leader election by erosion is compatible with the concurrency control framework, proving the existence of an amoebot algorithm that solves leader election under an unfair asynchronous adversary.
In order to apply the framework, we must first show that our algorithm satisfies the framework's three conventions:
\begin{enumerate}
    \item \textit{Validity.} For all system configurations in which an action $\alpha$ is enabled for an amoebot $A$, the execution of $\alpha$ by $A$ should be successful when all other amoebots are inactive.
    
    \item \textit{Phase Structure.} Each action should structure its operations as: (1) a ``compute phase'', during which an amoebot performs a finite amount of computation and a finite sequence of \Connected, \Read, and \Write\ operations, and (2) a ``move phase'', during which an amoebot performs at most one movement operation decided upon in the compute phase.
    
    \item \textit{Monotonicity.} Each action should be ``monotonic'', a lengthy technical condition describing actions' resilience to concurrent movements.
    Because our algorithm is \textit{stationary}, not involving any movement operations, we omit this definition.
\end{enumerate}

\begin{lemma} \label{lem:conventions}
    Algorithm~\ref{alg:erosion} satisfies the validity, phase structure, and monotonicity conventions of the concurrency control framework for amoebot algorithms.
\end{lemma}
\begin{proof}
    Actions in Algorithm~\ref{alg:erosion} only use \Connected, \Read, and \Write\ operations.
    The validity convention requires that every isolated execution of an enabled action succeeds.
    \Connected\ operations never fail.
    \Read\ and \Write\ operations only fail when the neighbor whose memory is being accessed disconnects from the calling amoebot due to a movement, which never happens in our algorithm because it is stationary.
    Thus, all executions of enabled actions in our algorithm succeed, satisfying validity.
    Phase structure can easily be verified by observing that---since our algorithm is stationary---its actions are entirely composed of compute phases.
    Finally, as observed in the concurrency control framework publication~\cite{Daymude2021-canonicalamoebot}, all stationary algorithms are trivially monotonic.
\end{proof}

Theorem 7 of~\cite{Daymude2021-canonicalamoebot} states that if an algorithm satisfies the framework's three conventions and every sequential execution of the algorithm terminates, then every asynchronous execution of the corresponding algorithm produced by the framework terminates in a configuration that was also reachable by the original algorithm.
Lemma~\ref{lem:conventions} shows that our algorithm for leader election satisfies the framework's three conventions, and Theorem~\ref{thm:seq3d} shows that every sequential execution of our algorithm terminates.
Therefore, we have the following corollary.

\begin{corollary} \label{cor:async}
    There exists an algorithm that solves the leader election problem for any connected and contractible amoebot system under the assumptions of 2D or 3D geometric space, assorted orientations, constant-size memory, and an unfair asynchronous adversary.
\end{corollary}

We note that the runtime of this asynchronous algorithm depends on the overhead of the concurrency control framework which has not yet been analyzed.

\section{Conclusion} \label{sec:conclude}

In this work, we presented the 3D geometric space variant for the amoebot model, extending the well-established model of programmable matter to 3D space.
We then presented a deterministic distributed algorithm for electing exactly one leader in $\bigo{n}$ rounds under an unfair sequential adversary (Theorem~\ref{thm:seq3d}) for connected and contractible systems, extending and simplifying the erosion-based approach of Di Luna et al.~\cite{DiLuna2020-shapeformation}.
This algorithm can be flexibly applied to both 2D and 3D space (Corollary~\ref{cor:seq2d}) and improves over the related algorithm for 3D leader election by Gastineau et al.~\cite{Gastineau2022-leaderelection} by extending the set of amoebot orientations and replacing the assumption of 2-neighborhood vision with local comparisons between 1-neighbors' orientations.
Finally, we proved that this algorithm can be transformed using the concurrency control framework for amoebot algorithms~\cite{Daymude2021-canonicalamoebot} to obtain the first known amoebot algorithm that solves leader election under an unfair asynchronous adversary (Corollary~\ref{cor:async}).

Although our leader election algorithm is the first to make use of the concurrency control framework for asynchronous correctness, several others discussed in Section~\ref{sec:problem} could likely do the same.
The proof of Lemma~\ref{lem:conventions} suggests that any stationary algorithm adhering to the standard amoebot assumptions (e.g., 1-neighborhood vision) will likely satisfy the framework's three conventions, and all of the deterministic algorithms already have proofs of termination under a sequential adversary or some weaker scheduler.
Some work is required to translate existing algorithms into the action-based semantics used by the canonical amoebot model and to reprove their correctness with respect to the new formulation, but the framework's subsequent application seems straightforward.

One potentially interesting extension of this work is to modify the leader election algorithm such that when it is run on an initially connected amoebot system that \textit{contains holes}, the existence of holes can be identified in a distributed way.
Just as the topological definition of holes identifies the inability to shrink the object to a single point, so too could the inability for our algorithm to erode a system to a single amoebot be a clue about the existence of holes.
Our algorithm already confirms the absence of holes as this is a necessary condition for a leader to emerge, but perhaps could be extended---likely with additional amoebot communication---to decide when holes exist.

Finally, we showed that our algorithm for leader election in 3D geometric space could be immediately applied without modification to a 2D geometric system as a special case (Corollary~\ref{cor:seq2d}).
This may not necessarily be true of every problem and algorithm.
What existing algorithms for 2D geometric systems can be viewed as special cases of generalized 3D algorithms?
What fundamental characteristics of a given algorithm or problem inhibit this translation?
Answers to these questions will accelerate the development of 3D amoebot algorithms and their potential applications to amoebots' modular robotic counterparts.

\bibliographystyle{plainurl}
\bibliography{ref}

\end{document}